\documentclass[conference, letterpaper, 10pt, twocolumn]{IEEEtran}

\IEEEoverridecommandlockouts

\usepackage{amsmath,amsfonts}
\usepackage{amsthm}  {
      \theoremstyle{plain}
			\newtheorem{definition}{Definition}
			\newtheorem{lemma}{Lemma}
			\newtheorem{theorem}{Theorem}
			\newtheorem{assumption}{Assumption}
			\newtheorem{example}{Example}
			\newtheorem{corollary}{Corollary}
			\newtheorem{proposition}{Proposition}
}

\usepackage[shortlabels]{enumitem}
\usepackage{graphicx}
\usepackage{booktabs}

\newcommand{\Ti}{\mathcal{T}} 
\newcommand{\St}{\mathcal{S}} 
\newcommand{\A}{\mathcal{A}} 
\newcommand{\Z}{\mathcal{Z}} 
\newcommand{\T}{\mathbb{T}} 
\newcommand{\Ob}{\mathbb{O}} 
\newcommand{\B}{\mathcal{B}} 
\newcommand{\R}{\mathbb{R}} 
\newcommand{\E}{{\mathbb{E}}} 

\DeclareMathOperator*{\argmax}{argmax}
\newcommand{\TODO}[1]{{\color{red}#1}}

\usepackage{algorithm}
\usepackage{algpseudocode}

\usepackage{url}
\usepackage[colorlinks=true,citecolor=red,pdftex]{hyperref} 

\begin{document}

\title{\LARGE \bf Myopic Policy Bounds for Information Acquisition POMDPs}

\author{Mikko Lauri, Nikolay Atanasov, George J. Pappas, and Risto Ritala
\thanks{M. Lauri and R. Ritala are with Department of Automation Science and Engineering, Tampere University of Technology, P.O. Box 692, FI-33101, Tampere, Finland, {\tt\small\{mikko.lauri, risto.ritala\}@tut.fi}.}%
\thanks{N. Atanasov and G. Pappas are with GRASP Lab, University of Pennsylvania, Philadelphia, PA 19104, USA, {\tt\small\{atanasov, pappasg\}@seas.upenn.edu}. This work was supported by TerraSwarm, one of six centers of STARnet, a Semiconductor Research Corporation program sponsored by MARCO and DARPA.}%
}
\maketitle

\begin{abstract}
This paper addresses the problem of optimal control of robotic sensing systems aimed at autonomous information gathering in scenarios such as environmental monitoring, search and rescue, and surveillance and reconnaissance. The information gathering problem is formulated as a partially observable Markov decision process (POMDP) with a reward function that captures uncertainty reduction. Unlike the classical POMDP formulation, the resulting reward structure is nonlinear in the belief state and the traditional approaches do not apply directly. Instead of developing a new approximation algorithm, we show that if attention is restricted to a class of problems with certain structural properties, one can derive (often tight) upper and lower bounds on the optimal policy via an efficient myopic computation. These policy bounds can be applied in conjunction with an online branch-and-bound algorithm to accelerate the computation of the optimal policy. We obtain informative lower and upper policy bounds with low computational effort in a target tracking domain. The performance of branch-and-bounding is demonstrated and compared with exact value iteration.
\end{abstract}

\section{Introduction}
\label{sec:intro}
The current proliferation of sensors and robots has potential to transform fields as diverse as environmental monitoring, security and surveillance, localization and mapping, and structure inspection. 
One of the technical challenges in these scenarios is to control the sensors and robots to extract accurate information about various physical phenomena autonomously. 

Robotic information acquisition may be thought of as an optimization problem: given constraints on robot motion and the available sensing resources, find an optimal control policy for the robot resulting in the greatest amout of collected information. 
Assuming Markovian system dynamics and observations conditionally independent given the system state, the optimization problem is formalized as a partially observable Markov decision process (POMDP) \cite{Kaelbling1998}.  
A POMDP consists of the state, action and observation spaces, stochastic system dynamics and observation models and a real-valued reward function. 
As the true state of the robot and environment are unknown, a probability density function (pdf) known as a belief state is maintained to represent information regarding the state. 
The optimal control policy is a mapping from belief states to actions such that the expected total reward accumulated over a given optimization horizon is maximized.

\begin{figure}[t!]
\centering
\includegraphics[width=\columnwidth]{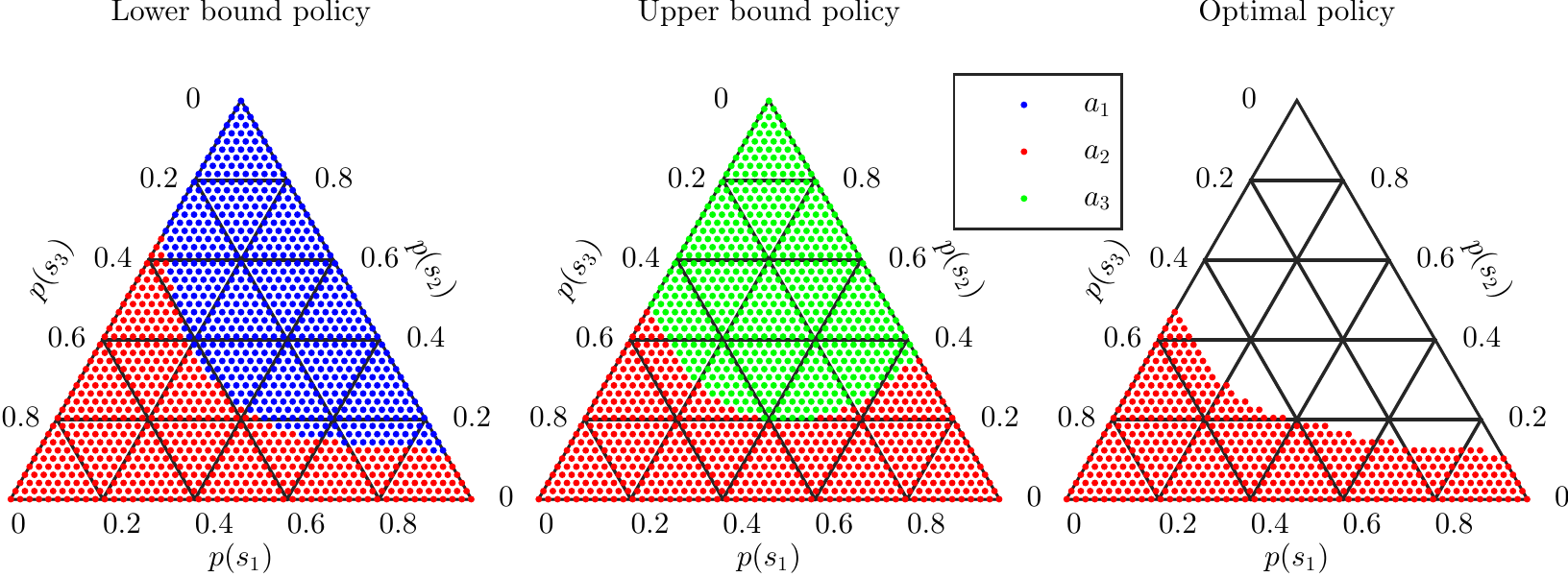}
\caption{The optimal policy $\mu^*$ in a POMDP (shown on the right) can be lower and upper bounded by myopic lower and upper bound policies $\underline{\mu}^*_1$ and $\overline{\mu}^*_1$ (shown on the left and middle, respectively). The ternary plots cover the belief simplex $\B$ in the three-state case with state space $\St = \{s_1, s_2, s_3 \}$, while the dots indicate the action selected by the policy for a particular belief point. The inequalities $\underline{\mu}^*_1(b) \leq \mu^*(b) \leq \overline{\mu}^*_1$ hold for all $b\in\B$. The optimal policy is fully determined if $\underline{\mu}^*_1(b) = \overline{\mu}^*_1(b)$. Even if the bounds do not agree they may provide information regarding the optimal action.}
\label{fig:policy_bound_example}
\end{figure}

While POMDPs offer a principled framework handling probabilistic uncertainty in both control actions and sensing outcomes, obtaining optimal policies for them is computationally hard~\cite{Madani2003}.
A wide range of algorithms for computing approximately optimal policies for POMDPs have been proposed in the literature, see e.g.\ \cite{Lovejoy1991,Hauskrecht2000,Shani2013} for reviews. 
Most research, however, focuses on a standard formulation in which the reward is a function of the hidden system state, and hence necessarily linear in the belief state. 
When it comes to information gathering applications in robotics (see e.g.\ \cite{Stachniss2005,Charrow2014,Atanasov2014,Lauri2015}), reward functions that capture uncertainty reduction (e.g.\, negative Shannon entropy~\cite{InfoTheoryBook}) or information gain (e.g., mutual information~\cite{InfoTheoryBook}) are often more interesting.
These reward function are nonlinear in the belief state, and hence, many of the existing approximation techniques for the standard POMDP formulation do not apply.
Approximation methods applicable for such information acquisition POMDPs include some online algorithms \cite{Ross2008} and specialized value iteration techniques, see e.g.\ \cite{Araya2010}.
Sampling-based online planning algorithms, such as those based on variants of Monte Carlo tree search (see e.g.\ \cite{Silver2010,Somani2013}), are also applicable.


Besides approximation algorithms, an alternative route to handling the inherent complexity of POMDPs with nonlinear belief-dependent rewards is to make structural statements about the optimal value function. 
In the case of state-dependent rewards, earlier work \cite{Lovejoy1987,Rieder1991,Krishnamurthy2015} shows that, if we restrict attention to a class of POMDPs that satisfy certain assumptions on the state and observation processes, we can bound the optimal policy from both below and above by an easily computable myopic (greedy) policy (see Fig.~\ref{fig:policy_bound_example} for details). 
The specific case of POMDP multi-armed bandits was studied in \cite{Krishnamurthy2009}.


The goal of this work is to extend these structural results to POMDPs with belief-dependent rewards. 
Krishnamurthy and Djonin~\cite{Krishnamurthy2007} determined conditions under which the optimal policy has a threshold structure with respect to the monotone likelihood ratio (MLR) order, a partial order on the belief states. 
Such partial orders on the belief states have been used in the related active hypothesis testing field~\cite{Naghshvar2013} to determine when a given sensing action dominates all other actions. 
For example, Naghshvar and Javidi~\cite{Naghshvar2010} used Blackwell ordering of the observation models to reduce an active hypothesis test with $K$ actions to a passive test with a single action. 
Compared to active hypothesis testing, information acquisition in POMDPs is even more challenging because the true underlying state evolves during the sensing process. 

Our first contribution is to extend the structural results of \cite{Krishnamurthy2015} for POMDPs with state-dependent rewards to POMDPs with nonlinear belief-dependent rewards. 
This allows us to generate upper and lower bounds on the optimal policy (which agree in large portions of the belief state space and hence determine the optimal action) via a very efficient myopic policy computation. 
Our second contribution is to apply the myopic policy bounds in conjunction with an online branch-and-bound pruning algorithm to accelerate computation of the optimal policy.
Our approach allows anytime computation with a bounded error from the optimal solution, unlike for instance sampling-based planning which is only asymptotically optimal.

The remainder of the paper is organized as follows.
Section~\ref{sec:problem} formulates the information acquisition problem as a POMDP, and discusses suitable choices of reward functions.
Section~\ref{sec:myopic} briefly reviews existing structural results. 
Section~\ref{sec:mlr_rewards} introduces our extension to belief-dependent rewards and defines the myopic lower and upper bound policies.
In Section~\ref{sec:bb}, we show how to apply the myopic policy bounds in a branch-and-bound pruning algorithm.
Section~\ref{sec:empirical} presents empirical evaluation in target tracking domains, and Section~\ref{sec:conclusion} concludes the paper.

\section{Information acquisition in POMDPs}
\label{sec:problem}
To model information acquisition problems in robotic applications we define a POMDP with a belief-dependent reward function. 
In detail, a POMDP is a tuple $\langle \Ti, \St, \A, \Z, \T, \Ob, \rho \rangle$, where $\Ti$ is a set of decision epochs and $\St = \{1, 2, \ldots, S\}$, $\A = \{1, 2, \ldots, A\}$, and  $\Z = \{1,2, \ldots, Z\}$ are the finite state, action, and observation spaces, respectively.
The function $\T:\St \times \St \times \A \to [0,1]$ is a stochastic state transition model such that $\T(s',s,a)$ is the probability of reaching state $s'\in\St$ from state $s \in \St$ after executing action $a\in\A$ and $\Ob: \Z \times \St \times \A \to [0,1]$ is a stochastic observation model such that $\Ob(z',s',a)$ is the conditional probability of observing $z'\in Z$ when state $s'\in\St$ was reached after executing action $a\in\A$. 
To simplify notation, let $\T^a := [\T(j,i,a)]_{i,j=1}^S$ be the matrix of state transition probabilities for a given action $a \in \A$, i.e.\ previous states are column-wise and next states row-wise. 
Similarly, let $\Ob^a := [\Ob(z,j,a)]$ for $z \in \Z$, $j \in \St$, and $a \in \A$.

As information about the system state is incomplete, it is modeled by a probability density function (pdf) $b(s), s\in \St$ over the system state. 
The set $\B$, called the belief space, is the set of all possible pdfs over the system state. 
Finally, $\rho:\B\times\A\to \R$ is the expected reward, which is a (nonlinear) function of the belief state $b$.

Our definition subsumes the traditional POMDP definition with a state-and-action-dependent reward function $r:\St\times\A\to\R$.
This is seen defining $\rho(b,a) := \sum_{s\in\St} r(s,a) b(s)$.


The evolution of the belief state $b$ in a POMDP is conditional on the actions and observations. When action $a \in \A$ is executed, the belief state evolves according to the state transition model $\T$. The evolution can be tracked by Bayesian filtering, which consists of iterating prediction and update steps. The prediction step revises the current belief state $b\in\B$ to the predicted belief $b^{a}\in\B$ according to
\begin{equation}
\label{eq:predicted_belief}
b^a(s') = \sum\limits_{s\in\St} \T(s',s,a) b(s).
\end{equation}
After the state transition, an observation $z'\in\Z$ is perceived and the information it provides is incorporated via the update step:
\begin{equation}
\label{eq:updated_belief}
b'(s') = \tau(b,a,z') := \frac{1}{\eta(z'\mid b,a)} \cdot \Ob(z',s',a) b^a(s'),
\end{equation}
where $\eta(z'\mid b,a) := \sum_{s \in \St} \Ob(z',s,a) b^a(s)$ is a normalization factor equal to the conditional probability of observing $z'$.

\subsection{Reward functions for information acquisition}
The goal in information acquisition is to reduce the uncertainty in the unobservable state of interest $s \in \St$. 
It is therefore desirable to reach ``peaked'' belief states that have the majority of their probability mass on a single underlying state.
In other words, actions that lead the belief state towards the vertices of the probability simplex $\B$ should be rewarded. 
Suitable measures of uncertainty are concave functions of the belief state, called uncertainty functions \cite[Sec.~14.16]{DeGroot2004}.
\begin{definition}[Uncertainty function and information gain \cite{DeGroot2004}]
An uncertainty function is a non-negative, concave function $f:\B\to\R^+$.
The information gain $\mathbf{I}_f(b, a)$ of an action $a\in\A$ in belief state $b\in\B$ is the expected reduction in the uncertainty function $f$:
\begin{equation}
\mathbf{I}_f(b, a) := f(b^a) - \E_Z[ f(\tau(b,a,Z)) ].
\end{equation}
\end{definition}
Examples of uncertainty functions include Shannon entropy, the more-general R{\'e}nyi quadratic entropy, and variants of the value of information, such as the probability of error in hypothesis testing~\cite{InfoTheoryBook}. For example, the information gain associated with Shannon entropy is known as mutual information.

For information acquisition tasks, either information gain or a negative uncertainty function may be used as the reward function in the POMDP formulation. However, as the following example shows, using an uncertainty function might be more appropriate when the magnitude of the predicted uncertainty $f(b^a)$ is significantly affected by the action choice.

\begin{example}
\label{xm:uncertainty_vs_infogain}
Consider an active localization problem~\TODO{\cite{Thrun2006}} in which a mobile robot needs to choose an appropriate action $a$ to reduce the uncertainty in the distribution $b$ of its current position $s$. Suppose that the entropy $f(b)$ of its current position distribution is $n$ bits. The robot may choose between a risky high-velocity motion $a_r$ leading to a predicted entropy of $f(b^{a_r}) = n+m$ bits or a safe low-velocity motion $a_s$ leading to $f(b^{a_s}) = n+h$ bits, where $h < m$. Suppose that in both cases, after moving, the robot receives the same amount of information, $k$ bits, from its sensor measurement $z'$. In this example, the negative uncertainty function correctly predicts that action $a_r$ is risky:
\begin{align*}
-f(\tau(b,a_r,z') ) &= -(n + m - k) \\
&< -(n + h - k) = -f(\tau(b,a_s,z) )
\end{align*}
but, perhaps surprisingly, the information gain prefers the risky action:
\[
\mathbf{I}_f(b,a_r) = m-k > h-k = \mathbf{I}_f(b,a_s).
\]
\end{example}

In the rest of the paper, we restrict attention to belief-dependent reward functions, specific to the task of information gathering, which have the following form:
\begin{equation}
\label{eq:info_reward}
\rho(b,a) := \sum_{s \in \St} r(s,a) b(s) - w_a f(b),
\end{equation}
where $f$ is an uncertainty function (preferred to information gain due to the observations in Example \ref{xm:uncertainty_vs_infogain}), $r(s,a)$ is any standard state-dependent reward function, and $w_a$ are user-specified weights that trade-off measurement rewards and state uncertainty.

\subsection{Value functions and optimal policies}
Let $H \in \mathbb{N}\cup\{\infty\}$ denote a planning horizon so that $\Ti := \{1,2,\ldots,H\}$, and let $0\leq \gamma \leq 1$\footnote{Note that $\gamma=1$ is only valid for finite $H$.} denote a discount factor determining the relative value of immediate and future rewards.
The goal in information acquisition is to choose a policy $\mu^*_k:\B \rightarrow \A$ for each $k \in \Ti$ such that the expected sum of rewards over the decisions epochs is maximized.
The sequence of optimal policies $\mu^*_k$ for $k\in\Ti$ remaining decisions can be computed via value iteration~\cite{Bertsekas1996} according to:
\begin{align}
Q_1^*(b,a) &= \rho(b,a)\\
\label{eq:value_iter}
V_k^*(b) &= \max\limits_{a\in\A} Q_k^*(b,a), \\
\label{eq:q_value}
Q_k^*(b,a) &= \rho(b,a) + \gamma \sum\limits_{z'\in\Z} \eta(z'\mid b,a) V_{k-1}^*(\tau(b,a,z')),\\
\label{eq:opt_policy}
\mu_k^*(b) &= \argmax\limits_{a\in\A} Q_k^*(b,a),
\end{align}
where $V_k^*:\B \to \R$ is the optimal value function. For infinite-horizon problems with $H\to\infty$ and bounded rewards, value iteration converges to a unique fixed point $V^*$ that satisfies the Bellman equation:
\begin{equation}
\label{eq:bellman_fixedpoint}
V^*(b) = \rho(b,a) + \gamma \sum\limits_{z'\in\Z} \eta(z'\mid b,a) V^*(\tau(b,a,z')),
\end{equation}
and the corresponding optimal infinite-horizon policy $\mu^*$ is stationary \cite{Bertsekas1996}.
As the set $\B$ of belief states over which the value iteration must be solved is uncountable, the basic scheme above rarely results in practical methods for computing optimal policies.

\section{Myopic policy bounds for POMDPs}
\label{sec:myopic}
In this section, we review existing results on myopic policy bounds for POMDPs with state-depend rewards. The idea of a myopic policy bound is that the optimal policy $\mu_k^*$ for any $k$ can be either \emph{lower or upper bounded} by the myopic policy $\mu_1^*$ (Eq.~\ref{eq:opt_policy}).
The results exploit concepts of stochastic partial orders, summarized in the following definition.

\begin{definition}[Stochastic partial orders \cite{Shaked2007}]
\label{def:stoch_orders}
Let $b_1, b_2 \in \B$.
\begin{enumerate}
\item $b_1$ first-order stochastically domainates $b_2$, denoted $b_1 \geq_s b_2$, if $\sum\limits_{i=j}^{S}b_1(i) \geq \sum\limits_{i=j}^{S} b_2(i)$ $\forall j \in \St$. 
\item Equivalently, $b_1 \geq_s b_2$ iff $\sum\limits_{i=1}^{S}b_1(i)g(i) \geq \sum\limits_{i=1}^{S}b_2(i)g(i)$ for all increasing functions $g:\St\to\R$.
\item $b_1$ is greater than $b_2$ in the monotone likelihood ratio (MLR) order, denoted $b_1 \geq_r b_2$, if $b_1(i)b_2(j) \leq b_2(i)b_1(j)$, for all $i<j$, $j\in\St$.
\item A function $g:\B\to\R$ is MLR increasing (decreasing) if $b_1 \geq_r b_2 \Rightarrow g(b_1) \geq g(b_2)$ $(g(b_1) \leq g(b_2))$.
\item MLR dominance implies first-order stochastic dominance, $b_1 \geq_r b_2 \Rightarrow b_1 \geq_s b_2$.
\end{enumerate}
\end{definition}

Given conditions on the state transition and observation models, it was shown in \cite{Lovejoy1987,Rieder1991} that the optimal value function is MLR increasing and that the optimal policy may be lower bounded by $\mu_1^*$.
The required conditions, however, were thought to be too restrictive to be of practical significance until a recent extension~\cite{Krishnamurthy2015} proposed the following revised assumptions.
\begin{assumption}[Sufficient conditions for existence of myopic policy bounds~\cite{Krishnamurthy2015}]
\label{ass:model_assumptions}
$ $
\begin{enumerate}[label=(A\arabic*)]
\item \label{asm:tp2} $\T^a$ and $\Ob^a$ are totally positive or order 2 (TP2) $\forall a\in\A$,
\item \label{asm:copositivity} For all $j \in \{1, 2, \ldots, S-1\}$, $a\in \{1, 2, \ldots, A-1\}$, and $z \in \Z$, the matrices $D^{j,a,z} := [d_{m,n}^{j,a,z} + d_{n,m}^{j,a,z}] \in \mathbb{R}^{S \times S}$ are copositive, where $d_{m,n}^{j,a,z} := \Ob^a_{z,j}\Ob^{a+1}_{z,j+1}\T^a_{m,j}\T^{a+1}_{n,j+1} - \Ob^a_{z,j+1}\Ob^{a+1}_{z,j}\T^a_{m,j+1}\T^{a+1}_{n,j}$, and
\item \label{asm:sum} $\sum\limits_{z\leq \bar{z}} \sum\limits_{j\in\St} \left[ \T^a_{i,j}\Ob^a_{z,j} - \T^{a+1}_{i,j}\Ob^{a+1}_{z,j}\right]$
$\geq 0, \forall i\in \St, \forall \bar{z} \in \Z$, $\forall a\in\{1,2,\ldots,A-1\}$.
\end{enumerate}
\end{assumption}
The TP2 property for a matrix requires that all its second-order minors are nonnegative~\cite{Fallat2011}. 
Examples of TP2 observation models include exponential family distributions such as Gaussian, exponential, gamma, and binomial pdfs~\cite{Fallat2011}; more examples may be found e.g.\ in \cite{Karlin1968}. 
While~\ref{asm:copositivity} may be checked directly~\cite{copositive_check}, it is much more efficient (but also more restrictive) to check the following:
\begin{enumerate}[leftmargin=2.5eM]
  \item[(A2')] \label{asm:copositivity_relaxed} $d_{m,n}^{j,a,z} + d_{n,m}^{j,a,z} \geq 0$ for all $m, n \in \St$, $j \in \{1, \ldots, S-1\}$, $a\in \{1, 2, \ldots, A-1\}$, and $z \in \Z$.
\end{enumerate}
Assumption~\ref{asm:sum} is a revised version of the original assumption in \cite{Lovejoy1987,Rieder1991}.
It states that for every state $i\in\St$, the prior pdf of perceiving $z\in\Z$ after ending in state $j\in\St$ after executing action $a+1$ first order stochastically dominates that of action $a$.
The assumptions lead to the following lemma.

\begin{lemma}[\cite{Krishnamurthy2015}]
\label{lem:update_monotone}
Consider a POMDP $\langle \Ti, \St, \A, \Z, \T, \Ob, \rho \rangle$. Then for all $b\in\B$, and $a', a\in\A$, $a'\geq a$,
\begin{enumerate}[label=(L\arabic*)]
\item \label{lem:sdom} if \ref{asm:sum} holds, $\eta(z'\mid b, a') \geq_s \eta(z'\mid b, a)$, and
\item \label{lem:rdom} if \ref{asm:copositivity} holds, $\tau(b,a',z') \geq_r \tau(b,a,z')$ for all $z'\in\Z$.
\end{enumerate}
\end{lemma}

The following theorem for POMDPs with rewards linear in the belief state is due to \cite[Prop.~2]{Lovejoy1987} and \cite[Theorem~1]{Krishnamurthy2015}:
\begin{theorem}
\label{thm:mlr_policy_bounds}
Consider a POMDP $\langle \Ti, \St, \A, \Z, \T, \Ob, \rho \rangle$ satisfying Assumptions~\ref{asm:tp2}-\ref{asm:sum}, with $\rho(b,a) := \sum_{s\in\St} r(s,a) b(s)$.
If $r(s,a)$ is increasing (decreasing) in $\St$ for all $a\in\A$, then the optimal value function $V_k^*$ is MLR increasing (decreasing) and the optimal policy $\mu_k^*(b)$ for any $k\in\Ti$ satisfies $\mu_1^*(b) \leq \mu_k^*(b)$ ($\mu_1^*(b) \geq \mu_k^*(b)$) $\forall b\in\B$.
\end{theorem}

\section{Transformed MLR monotone rewards for information acquisition POMDPs}
\label{sec:mlr_rewards}
In this section, we extend Theorem~\ref{thm:mlr_policy_bounds} to apply to reward functions of the form in~\eqref{eq:info_reward}, which are of interest for information acquisition tasks. To proceed, we need the following minor extension of Theorem~\ref{thm:mlr_policy_bounds}.
\begin{corollary}
\label{cor:mlr_policy_bounds}
Consider a POMDP $\langle \Ti, \St, \A, \Z, \T, \Ob, \rho \rangle$ satisfying Assumptions~\ref{asm:tp2}-\ref{asm:sum} and suppose that $\rho$ is MLR increasing (decreasing) in $\B$ for all $a\in\A$. Then, the optimal policy $\mu_k^*(b)$ for any $k\in\Ti$ satisfies $\mu_1^*(b) \leq \mu_k^*(b)$ ($\mu_1^*(b) \geq \mu_k^*(b)$) $\forall b\in\B$.
\end{corollary}
\begin{proof}
By \cite[Prop.~2]{Lovejoy1987}, it suffices to show $V_k^*$ is MLR increasing.
We proceed by induction.
For the base case $k=1$, $V_1^*(b) = \max\limits_{a\in\A} \rho(b,a)$ is clearly MLR increasing.
For the induction step, suppose the claim holds for $V_{k-1}^*$. 
Consider the sum part of the value iteration in Eq.~\eqref{eq:q_value}.
Let $b_1, b_2\in\B$ s.t.\ $b_1 \geq_r b_2$.
Now
\begin{equation}
\begin{split}
&\sum\limits_{z'=1}^Z \eta(z'\mid b_1,a) V_{k-1}^*(\tau(b_1,a,z')) \\
\geq &\sum\limits_{z'=1}^Z \eta(z'\mid b_2,a) V_{k-1}^*(\tau(b_1,a,z'))\\
\geq &\sum\limits_{z'=1}^Z \eta(z'\mid b_2,a) V_{k-1}^*(\tau(b_2,a,z')).
\end{split}
\end{equation}
The first inequality follows by Lemma~\ref{lem:sdom}, \cite[Lemma~1.2(2),1.3(1)]{Lovejoy1987} which state $\tau(b,a,z')$ is MLR increasing in $z'$, and the the induction hypothesis.
The second inequality follows by Lemma~\ref{lem:rdom} and the induction hypothesis.
The proof for the MLR decreasing part is similar and omitted.
\end{proof}

Following the observation in Corollary~\ref{cor:mlr_policy_bounds}, the main idea is to transform the reward function $\rho(b,a)$~\eqref{eq:info_reward} to one that is MLR monotone but leaves the optimal policy unchanged. Accordingly, we define two transformed reward functions:
\begin{equation}
\label{eq:transformed_reward}
\begin{split}
\underline{\rho}(b,a) &= \rho(b,a) + [(I-\gamma(\T^a)^T)g]^T b\\
\overline{\rho}(b,a) &= \rho(b,a) + [(I-\gamma(\T^a)^T)h]^T b,
\end{split}
\end{equation}
where $g,h \in \mathbb{R}^S$ are parameter vectors chosen so that $\underline{\rho}(b,a)$ is MLR increasing and $\overline{\rho}(b,a)$ is MLR decreasing in $\B$.

Note that, by construction the transformed rewards in \eqref{eq:transformed_reward} leave the optimal policy unaffected. This can be verified by plugging either one into the Bellman equation \eqref{eq:bellman_fixedpoint}. The corresponding infinite horizon value functions are $\underline{V}^*(b) = V^*(b) + g^T b$, $\overline{V}^*(b) = V^*(b) + h^T b$, and $\underline{\mu}^* \equiv \mu^* \equiv \overline{\mu}^*$~\cite{Krishnamurthy2015}. The following theorem gives necessary conditions for the existence of vectors $g$ and $h$ at a single belief state.

\begin{theorem}
\label{thm:mlr_LP}
Let $b_0\in\B$. If there exists a $g\in\R^S$ such that the set of linear constraints given by
\begin{equation}
\label{eq:mlr_LP}
\left.\frac{\partial \rho(b,a)}{\partial b}\right|_{b=b_0}K + [(I-\gamma(\T^a)^T)g]^T K \geq 0 \quad \forall a \in \A
\end{equation}
where $K$ is an $S$-by-$S$ matrix with entries
\begin{equation}
k_{ij} = \begin{cases}
1 & \text{if } i=j\\
-1 & \text{if } j = i+1
\end{cases},
\end{equation}
are satisfied, then $\underline{\rho}(b,a)$ as defined in Eq.~\eqref{eq:transformed_reward} is MLR increasing at $b_0$. 
\end{theorem}
\begin{proof}
As in the supplementary material of \cite{Krishnamurthy2015}, let $\Delta = \{\delta \in \mathbb{R}^S: 1=\delta(1)\geq \delta(2) \geq \ldots \geq \delta(S) \}$.
Any belief state $b\in\B$ may be represented as $\delta = Kb \in \Delta$.
Let $b_1 = K\delta_1, b_2 = K\delta_2$ such that $b_1 \geq_r b_2$.
By Definition~\ref{def:stoch_orders}, $b_1 \geq_s b_2$, and the equivalent partial order on $\Delta \subset \mathbb{R}^S$ is the component-wise partial order between $\delta_1$ and $\delta_2$.
Then, the function $\underline{\rho}(K\delta,a)$ is increasing on $\Delta$ if its partial derivatives are non-negative:
\begin{equation}
\begin{split}
\frac{\partial \underline{\rho}(K\delta,a)}{\partial \delta} &= \frac{\partial \rho(K\delta,a)}{\partial \delta} + \frac{\partial [(I-\gamma(\T^a)^T)g]^T K\delta}{\partial \delta}\\
&=\frac{\partial \rho(b,a)}{\partial b}K + [(I-\gamma(\T^a)^T)g]^T K \geq 0
\end{split}
\end{equation}
and the claim follows by evaluating the derivative at $b_0$.
\end{proof}
The proof for $\overline{\rho}(b,a)$ being MLR decreasing is obtained similarly, by swapping the required sign for the partial derivatives. 
Practically, the condition of the theorem may be checked by solving a feasibility linear program (LP) with $(S-1)\cdot A$ constraints.
If a single $g$ and $h$ may be chosen such that the transformed rewards are MLR monotone for \emph{all} $b\in\B$, then the myopic policies
\begin{equation}
\begin{split}
\underline{\mu}_1^*(b) &= \argmax_{a\in\A} \underline{\rho}(b,a)\\
\overline{\mu}_1^*(b) &= \argmax_{a\in\A} \overline{\rho}(b,a)
\end{split}
\end{equation}
satisfy $\underline{\mu}_1^*(b) \leq \mu^*(b) \leq \overline{\mu}_1^*(b)$ for all $b\in\B$.
This follows from Corollary~\ref{cor:mlr_policy_bounds}, and the fact that the optimal stationary policy is unaffected by applying the transformed rewards.
It is imporant to note that the preceding inequalities do \emph{not} hold anymore for finite horizon optimal policies.

Consider the following two examples of nonlinear belief-dependent reward functions obtained by applying different uncertainty functions in Eq.~\eqref{eq:info_reward}. In both examples, we denote by $r_a$ the column vector formed from the elements of $r(s,a)$ for all $s\in\St$.

\begin{example}[Shannon Entropy]
Let the uncertainty function in~\eqref{eq:info_reward} be the Shannon entropy:
\[
f(b) := -\sum\limits_{i=1}^S b(i)\log b(i).
\]
The derivative of $\rho(b,a)$ for the constraints in Thm.~\ref{thm:mlr_LP} is:
\[
\frac{\partial \rho(b,a)}{\partial b} = r_a^T + w_a\left[\begin{matrix} 1+\log(b(1)), & \ldots & 1+\log(b(S))\end{matrix}\right].
\]
As $\log(0)$ is not well defined, we consider the inner belief simplex $\B_\epsilon := \{b\in\B \mid  \forall i :  b(i) \geq \epsilon\}$ for $\epsilon \in (0,1)$.
Note that 
\[
\log(\epsilon) \leq \log(b(i+1)) - \log(b(i)) \leq -\log(\epsilon).
\]
For $i = 1,\ldots,S-1$ and $\forall a\in\A$, the $i$th constraint in Eq.~\eqref{eq:mlr_LP} can then be bounded by a belief-independent quantity:
\[
\phi^a_{i} - \phi^a_{i+1} \geq w_a \left( \log(b(i+1) - \log(b(i)) \right) \geq -w_a\log(\epsilon),
\]
where $\phi^a := r_a +(I-\gamma(\T^a)^T)g$.
\end{example}
\begin{example}[R{\'e}nyi Quadratic Entropy]
\label{xm:quadratic_entropy}
Let the uncertainty function in~\eqref{eq:info_reward} be the R{\'e}nyi quadratic entropy:
\[
f(b) := -\log \sum_{i=1}^S b^2(i).
\]
The derivative of $\rho(b,a)$ for the constraints in Thm. \ref{thm:mlr_LP} is:
\[
\frac{\partial \rho(b,a)}{\partial b} = r_a^T + \frac{2 w_a}{\sum_{i=1}^S b^2(i)} b^T.
\]
Note that:
\[
-1 \leq b(i+1) - b(i) \leq 1 \qquad \text{and} \qquad \frac{1}{S} \leq \sum_{i=1}^S b^2(i) \leq 1,
\]
which can be used to obtain belief-independent bounds in \eqref{eq:mlr_LP}. In particular, the $i$th constraint in \eqref{eq:mlr_LP} may be bounded by a belief-independent quantity:
\[
\phi^a_i - \phi^a_{i+1} \geq \frac{2 w_a}{\sum_{j=1}^S b^2(j)} (b(i+1) - b(i))\geq -2 w_a,
\]
where $\phi^a := r_a +(I-\gamma(\T^a)^T)g$ and $i = 1,\ldots,S-1$.
\end{example}
In both examples, belief independent bounds for finding $h$ are obtained with straightforward changes.
For Shannon entropy, one ideally wants to choose $\epsilon$ as small as possible to cover the greatest possible subset of the belief space.
On the other hand, the smaller $\epsilon$ is, the more restrictive the constraints become, making it more unlikely to find feasible solutions.
The issue is resolved by considering a different quantification of uncertainty, such as the R{\'e}nyi quadratic entropy.

\begin{figure}[t!]
\centering
\includegraphics[width=\columnwidth]{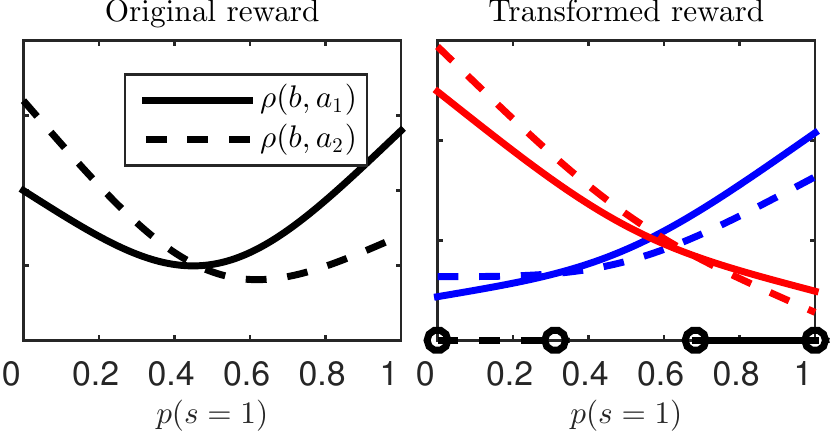}
\caption{Expected original (left axes) and transformed (right axes) rewards as function of the belief state according to Example~\ref{xm:quadratic_entropy}.
Rewards for action $a_1$ are shown by a solid line and for action $a_2$ by a dashed line.
The increasing blue curves in the right axes depict the MLR increasing rewards $\underline{\rho}(b,a)$, and the decreasing red curves depict the MLR decreasing rewards $\overline{\rho}(b,a)$.
The two regions on the right axes indicated by lines starting and ending with circle markers are the parts of the belief space where the optimal policy is fully determined by the myopic policy bounds.
The dashed region corresponds to $\underline{\mu}^*_1(b) = \overline{\mu}^*_1(b) = a_2$ and the solid region to $\underline{\mu}^*_1(b) = \overline{\mu}^*_1(b) = a_1$.}
\label{fig:reward_example}
\end{figure}

Figure~\ref{fig:reward_example} shows an example of the transformation of R{\'e}nyi quadratic entropy rewards and the associated monotonicity properties as discussed in Example~\ref{xm:quadratic_entropy}.
The left axes show the original expected reward as function of the belief state, and the right axes show the transformed MLR increasing and decreasing rewards, respectively.
For $S=2$, the belief space $\B$ can be represented by $[0,1]\subset \mathbb{R}$ and MLR monotonicity is equivalent to the familiar notion of monotonicity on $\mathbb{R}$.
On the bottom of the right axes, two regions in $\B$ are indicated by a dashed and a solid line.
These are regions where the myopic policy bounds agree and the optimal policy is thus fully determined.

\section{Myopic bounds for branch-and-bound pruning}
\label{sec:bb}
The significance of Theorems~\ref{thm:mlr_policy_bounds} and~\ref{thm:mlr_LP} is that they give a prescription for constructing policy bounds $\underline{\mu}_1^*(b) \leq \mu^*(b) \leq \overline{\mu}_1^*(b)$ for all $b\in\B$. In this section, we show that the bounds can be used in an online branch-and-bound scheme to accelerate the computation of the optimal policy for POMDPs that satisfy Assumption~\ref{ass:model_assumptions}.

Branch-and-bound pruning for POMDPs is based on a tree search over the set of belief states reachable by finite-length action and observation sequences, see e.g.\ \cite{Ross2008}. The search is started from the current belief state $b_0$ at the root of the search tree.
The tree branches due to action choices and observation possibilities, so that each node corresponds to a belief state $b$ reachable from $b_0$ via the prediction and update iteration in~\eqref{eq:predicted_belief},~\eqref{eq:updated_belief}. The standard branch-and-bound algorithm works with lower and upper bounds on the optimal action \emph{values}, i.e., $\underline{Q}_k(b,a) \leq Q^*_k(b,a) \leq \overline{Q}_k(b,a)$. Starting from node $b$, the subtree corresponding to action $a$ may be pruned if it is known to be suboptimal (i.e, the upper bound $\overline{Q}_k(b,a)$ is lower than the lower bound $\underline{Q}_k(b,\hat{a})$ for some other action $\hat{a}$), thus saving computational resources.

Following the intuition of the branch-and-bound scheme above, we design a new algorithm (Alg.~\ref{alg:bb}), which employs our policy bounds, \emph{instead of} bounds on the optimal action values. Given a belief state $b$ and a search depth $d>0$, the algorithm constructs a search tree over reachable belief states, storing the estimate of the value $\hat{Q}(b,a)$ for each visited belief state.
Only actions within the myopic lower and upper bound policies are considered.

\begin{algorithm}[ht]
\caption{Branch-and-bound with myopic policy bounds.}          
\label{alg:bb}                           
\begin{algorithmic}[1]
\Require Belief state $b$, search depth $d$
\Ensure Estimate of the value of the best action at $b$
\Function{Search}{$b$, $d$}
\If {$d = 1$} \Return $\max\limits_{a}\rho(b,a)$ \EndIf
\ForAll{$a \in \A \cap [\underline{\mu}^*_1(b), \overline{\mu}^*_1(b)]$} 
\State $\hat{Q}(b,a) \gets \rho(b,a) +\gamma \sum\limits_{z' \in \Z}\biggl(\eta(z'\mid b,a)$
\State \hfill$\times \text{\Call{Search}{$\tau(b,a,z')$, $d-1$}}\biggr)$
\EndFor
\State \Return $\max\limits_{a} \hat{Q}(b,a)$
\EndFunction
\end{algorithmic}
\end{algorithm}
The next proposition follows from \cite[Theorem~3.1]{Ross2008}.
\begin{proposition}
After Alg.~\ref{alg:bb} terminates, its action recommendation $\hat{a}$ is guaranteed to coincide with the optimal action $a^*:=\mu^*(b)$ if the myopic lower and upper bounds agree at $b$, i.e., if $\underline{\mu}^*_1(b) = \overline{\mu}^*_1(b)$. Otherwise, the suboptimality of $\hat{a}$ is of order $\gamma^d$, i.e., $|Q^*(b,\hat{a}) - Q^*(b,a^*)| \sim o(\gamma^d)$. 
\end{proposition}

\section{Application to Target Tracking}
\label{sec:empirical}
We evaluate the performance of Algorithm~\ref{alg:bb} in a target tracking application. Let $s \in \St$ model the state of a physical process of interest such as the position of an adversarial aircraft or certain environment conditions (e.g., temperature, chemical concentration, air pollution, etc.). Let $s=1$ correspond to an innocuous target state (e.g., the aircraft poses no danger or the environment levels are normal) and, at the other extreme, let $s=S$ correspond to a red-alert state in which there is imminent danger. We aim to design a control policy for an autonomous robotic platform to track the target state $s$. Let the available robot actions $\A=\{1, 2, \ldots, A\}$ model the priority invested into tracking the target, with $1$ corresponding to greatest priority (e.g, all available resources being dedicated to tracking the target) and $A$ corresponding to lowest effort, i.e.\ the target is not being tracked.

The state transition model is such that the lower the priority invested into tracking, the more likely it is to transition to a state that is considered dangerous.
Conversely, the higher the priority given to tracking, the less likely it is to accidentally enter a dangerous state.
In other words, supposing the current state is $i$, the likelihood of transitioning to state $j>i$ should increase as function of $a$.
State transition matrices $\T^a$ modeling this can be defined e.g. as follows.
Let $\{x_i\}_{i=1}^S$ and $\{y_j\}_{j=1}^{SA}$ be two increasing sequences of real numbers, and define a matrix $A_{ij}=\exp(x_iy_j)$, normalizing it such that each column sums to one, and define $\T^a$ as the submatrix of $A$ containing all rows and columns $j$ for which $1+(a-1)\cdot S \leq j \leq a\cdot S$.
Each $\T^a$ also satisfies the required TP2 condition \cite{Fallat2011}.

The robot has access to $Z=S$ observations providing information about the target state.
The observation matrix $\Ob^a$ for each action $a\in\A$ has elements
\begin{equation}
\label{eq:obsmodel_example}
\Ob^a_{z,j} = \begin{cases}
q & \textnormal{if } z=j\\
(1-q)/2 &\textnormal{if } j\neq 1, i\neq S \textnormal{ and } |z-j|=1\\
(1-q) &\textnormal{if } j=1, i=S \textnormal{ and } |z-j|=1\\
0 & \textnormal{otherwise}
\end{cases},
\end{equation}
i.e.\ with probability $q<1$ the observation identifies the true state, with symmetric error probability in either direction.
This model corresponds to a situation where e.g.\ entering a dangerous state is perceived e.g.\ by the robot being damaged by hazardous environmental conditions.
This class of target tracking domains can be checked to satisfy Assumptions~\ref{ass:model_assumptions}.

\subsection{Minimizing the uncertainty of the target state}
We first consider target tracking with a penalty for uncertain tracker state information.
We set $S=A=Z=3$ and defined the sequences $\{x_i\}=\{1,2,3\}$, $\{y_j\}=\{-9,-8,\ldots,-1\}$.
The observation model was as defined in Eq.~\eqref{eq:obsmodel_example}.
The reward was as in Eq.~\eqref{eq:info_reward}, with
\begin{equation}
r(s,a) = \left[\begin{matrix}
2 & 2.5 & 1\\
1.1 & 1.2 & 0.5\\
0.3 & 2 & 0.2
\end{matrix}\right],
\end{equation}
$w_a = \left[\begin{matrix}1.1 & 1.6 &1\end{matrix}\right]$, and $f$ equal to R{\'e}nyi quadratic entropy.
The discount factor was $\gamma = 0.99$, and the observation accuracy parameter was $q=0.7$.

We computed the lower and upper myopic policy bounds for the problem.
The bounds are visualized on the belief simplex $\B$ in Figure~\ref{fig:policy_bound_example}.
The figure shows on the left and middle the lower and upper myopic policies, and on the right the optimal policy determined when the two bounds agree.

We defined a larger problem with $S=Z=8$ and $A=3$, and in it compared Algorithm~\ref{alg:bb} against an exhaustive tree search.
There are $N_d = (AZ)^d$ reachable belief states after $d$ decisions, and the number of belief states in a complete search tree is $\sum_{i=0}^d N_i$.

\begin{table}[ht]
\centering
\caption{Average percentage of belief states pruned as function of the search depth $d$ applying the myopic policy bounds.}
\label{tab:expansions}
\begin{tabular}{@{}cccccc@{}}
\toprule
        			        & $d=2$ & $d=3$ & $d=4$ & $d=5$ & $d=6$ \\ \midrule
Belief states pruned & 26.3\%  & 36.3\%   & 44.9\%  & 52.3\%  & 58.8\%  \\
\bottomrule
\end{tabular}
\end{table}

Table~\ref{tab:expansions} shows the average percentage of belief state pruned from the search tree for 100 randomly sampled initial belief states.
Computing the myopic policy bound requires solving a feasibility LP with $(S-1)\cdot A$ constraints offline, and then during the search computing for each belief encountered the action maximizing the immediate expected transformed reward.
Empirically, we found that for $d\geq 3$ branch-and-bounding was faster than the exhaustive search measured by computation time.

\subsection{Comparison with existing approaches}
Next, we compare the performance of our approach to an optimal incremental pruning algorithm\footnote{We applied the implementation from the pomdp-solve package of A. Cassandra, see \url{http://pomdp.org}} in a target tracking scenario with a state space size $S$ between $4$ and $256$, and action space size $A$ between $4$ and $64$. 
Since the incremental pruning algorithm (as well as other existing approaches) can handle only state-dependent rewards, we do not add an uncertainty function to the reward, i.e., $f \equiv 0$ in~\eqref{eq:info_reward}.
The state-dependent reward $r(s,a)$ is designed to 1) incur a lower reward for investing higher priority in target tracking and to 2) penalize for both tracking the target poorly or too dangerously.
The penalty on dangerous tracking models a situation where attempting to track the state too aggressively may make the robot vulnerable to environmental hazards.
The reward is set as $r(s,a) = -c_a\cdot (A-a+1) + t(s,a)$, where $c_a$ is the cost of expending one unit of effort in tracking, and $t(s,a)$ is a tracking performance reward, defined
\begin{equation}
t(s,a) = \begin{cases}
-c_p\cdot \frac{1}{s} & \textnormal{if } s \in \St_p\\
-c_d\cdot \frac{1}{S-s+1} & \textnormal{if } s \in \St_d\\
k_r\cdot s & \textnormal{otherwise}
\end{cases}
\end{equation}
where $\St_p$ and $\St_d$ are the regions of the state space corresponding to poor and dangerous tracking, respectively, and $c_p$ and $c_d$ are the respective penalty costs.
If the state is not in $\St_l$ or $\St_h$, a tracking reward proportional to $k_r$ is received.

In our experiments, we set $c_a=\frac{1}{2A}$, $c_p=1$, $c_d=0.1$, and $k_r=\frac{2}{S}$.
Furthermore, we set $\St_l = \{s\in\St\mid s \leq S/5\}$ and $\St_h = \{s\in\St\mid s\geq9S/10\}$, and $\gamma=0.99$.
The observation accuracy parameter was $q=0.8$.

Table~\ref{tab:tracking_pruning} shows the minimum, average and maximum percentage of actions that could be pruned by applying the myopic policy bounds as function of the domain size.
A value of 0\% would indicate that the bounds were completely non-informative, and a value of 100\% that the optimal policy is fully determined by the bounds.
The results were obtained by evaluating the bounds with 500 randomly sampled reachable belief states in each problem.
We note that pruning efficiency tends to decrease both as function of the state space and action space size, but in all cases at least one third pruning rate was achieved even in the worst case.
For $S=256$ and $A=64$, a vector satisfying the constraints of Thm.~\ref{thm:mlr_LP} could not be found.

\begin{table*}[ht]
\centering
\caption{The (minimum, average, maximum) percentage of actions pruned in target tracking domains as function of state space size $S$ and action space size $A$.}
\label{tab:tracking_pruning}
\begin{tabular}{@{}cccccc@{}}
\toprule
$S$   & $A=4$ & $A=8$ & $A=16$ & $A=32$ & $A=64$ \\ \midrule
4   & (75.0, 94.9, 100)    & (75.0, 83.3, 87.5)    & (75.0, 78.7, 87.5)     & (71.9, 74.7, 75.0)     & (64.1, 72.0, 75.0)     \\
8   & (50.0, 75.0, 75)    & (37.5, 59.5, 62.5)    & (50.0, 58.0, 75.0)     & (59.4, 64.9, 71.9)     & (48.4, 57.6, 67.2)    \\ \midrule
16  & (50.0, 72.7, 75)    & (50, 65.2, 75.0)    & (43.8, 61.3, 68.8)     & (46.9, 55.0, 62.5)     & (46.9, 55.5, 62.5)    \\
32  & (50.0, 73.9, 75)    & (37.5, 58.5, 62.5)    & (37.8, 54.9, 62.5)     & (37.5, 49.9, 59.4)     & (43.8, 52.6, 59.4)     \\ \midrule
64  & (50.0, 71.9, 75)    & (37.5, 60.4, 62.5)    & (43.8, 54.0, 62.5)     & (37.5, 49.6, 59.4)     & (40.6, 47.9, 53.1)     \\
128 & (50.0, 73.2, 75)    & (50, 61.1, 62.5)    & (43.8, 54.7, 56.3)     & (40.6, 49.3, 53.1)     & (42.2, 48.3, 53.1)     \\ \midrule
256 & (50.0, 72.8, 75)    & (50, 60.4, 62.5)    & (43.8, 50.9, 56.3)     & (43.8, 49.1, 53.1)     & (n/a)     \\
    \bottomrule
\end{tabular}
\end{table*}

\begin{figure}[thpb]
\centering
\includegraphics[width=\columnwidth]{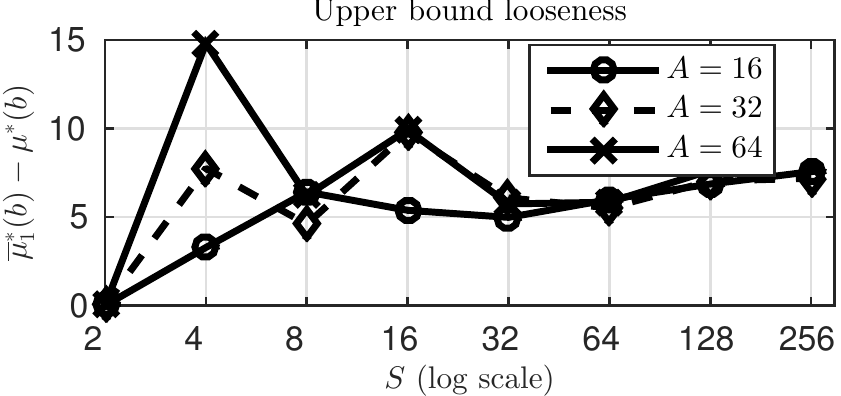}
\caption{The average distance between the upper bound and optimal policy as function of state space size $S$ for various action space sizes $A$.}
\label{fig:looseness}
\end{figure}

As we computed the optimal solution, we could examine how close the bounds are to the optimal policy.
Figure~\ref{fig:looseness} shows an example of the looseness of the upper bound measured by the average of $\overline{\mu}^*_1(b)-\mu^*(b)$ over the 500 belief states.
The greater the difference is, the further away the myopic upper bound policy is from the optimal policy, and the less informative it is.
We note that even as the number of states and actions increases, the upper bound remains at a distance of less than 10 actions from the optimal.

\section{Conclusion}
\label{sec:conclusion}
We examined information acquisition in POMDPs with a reward function nonlinear in the belief state.
We showed that if the POMDP fulfills certain structural properties, the optimal infinite horizon stationary policy may be lower and upper bounded by myopic (greedy) lower and upper bound policies.
Based on the bounds, we designed a branch-and-bound pruning algorithm for online planning in POMDPs, and demonstrated its effectiveness in a target tracking application.

The main advantage of our approach is that, if the structural properties are satisified in the POMDP model, the resulting branch-and-bound algorithm saves orders of magnitude in computation compared to the existing approaches. 
The main drawback is that the requred structural properties can be quite restrictive in some applications. 
For example, a natural order in the state and action spaces is required.

Future work will focus on relaxing the structural requirements while keeping track of the effect on the policy bounds. 
This has potential to widen the scope of the possible applications of our apporach significantly.




\bibliographystyle{IEEEtran}
\bibliography{ref}

\begin{thebibliography}{10}
\providecommand{\url}[1]{#1}
\csname url@samestyle\endcsname
\providecommand{\newblock}{\relax}
\providecommand{\bibinfo}[2]{#2}
\providecommand{\BIBentrySTDinterwordspacing}{\spaceskip=0pt\relax}
\providecommand{\BIBentryALTinterwordstretchfactor}{4}
\providecommand{\BIBentryALTinterwordspacing}{\spaceskip=\fontdimen2\font plus
\BIBentryALTinterwordstretchfactor\fontdimen3\font minus
  \fontdimen4\font\relax}
\providecommand{\BIBforeignlanguage}[2]{{%
\expandafter\ifx\csname l@#1\endcsname\relax
\typeout{** WARNING: IEEEtran.bst: No hyphenation pattern has been}%
\typeout{** loaded for the language `#1'. Using the pattern for}%
\typeout{** the default language instead.}%
\else
\language=\csname l@#1\endcsname
\fi
#2}}
\providecommand{\BIBdecl}{\relax}
\BIBdecl

\bibitem{Kaelbling1998}
\BIBentryALTinterwordspacing
L.~Kaelbling, M.~Littman, and A.~Cassandra, ``{Planning and acting in partially
  observable stochastic domains},'' \emph{Artificial Intelligence}, vol. 101,
  no. 1-2, pp. 99--134, 1998. [Online]. Available:
  \url{http://linkinghub.elsevier.com/retrieve/pii/S000437029800023X}
\BIBentrySTDinterwordspacing

\bibitem{Madani2003}
\BIBentryALTinterwordspacing
O.~Madani, S.~Hanks, and A.~Condon, ``{On the Undecidability of Probabilistic
  Planning and Related Stochastic Optimization Problems},'' \emph{Artificial
  Intelligence}, vol. 147, no. 1--2, pp. 5--34, 2003. [Online]. Available:
  \url{http://www.sciencedirect.com/science/article/pii/S0004370202003788}
\BIBentrySTDinterwordspacing

\bibitem{Lovejoy1991}
W.~Lovejoy, ``{A Survey of Algorithmic Methods for Partially Observed Markov
  Decision Processes},'' \emph{Annals of Operations Research}, vol.~28, no.~1,
  pp. 47--65, 1991.

\bibitem{Hauskrecht2000}
M.~Hauskrecht, ``{Value-function Approximations for Partially Observable Markov
  Decision Processes},'' \emph{Journal of Artificial Intelligence Research},
  vol.~13, no.~1, pp. 33--94, 2000.

\bibitem{Shani2013}
G.~Shani, J.~Pineau, and R.~Kaplow, ``{A Survey of Point-based POMDP
  Solvers},'' \emph{Autonomous Agents and Multi-Agent Systems}, vol.~27, no.~1,
  pp. 1--51, 2013.

\bibitem{Stachniss2005}
C.~Stachniss, G.~Grisetti, and W.~Burgard, ``{Information Gain-based
  Exploration using Rao-Blackwellized Particle Filters},'' in \emph{Proc.
  Robotics: Science and Systems (RSS)}, Cambridge, MA, USA, Jun. 2005.

\bibitem{Charrow2014}
B.~Charrow, V.~Kumar, and N.~Michael, ``{Approximate Representations for
  Multi-robot Control Policies that Maximize Mutual Information},''
  \emph{Autonomous Robots}, vol.~37, no.~4, pp. 383--400, Aug. 2014.

\bibitem{Atanasov2014}
N.~Atanasov, J.~{Le Ny}, K.~Daniilidis, and G.~J. Pappas, ``{Information
  Acquisition with Sensing Robots: Algorithms and Error Bounds},'' in
  \emph{Proc. IEEE International Conference on Robotics and Automation (ICRA)},
  Hong Kong, China, Jun. 2014, pp. 6447--6454.

\bibitem{Lauri2015}
M.~Lauri and R.~Ritala, ``{Optimal Sensing via Multi-armed Bandit Relaxations
  in Mixed Observability Domains},'' in \emph{Proc. IEEE International
  Conference on Robotics and Automation (ICRA)}, Seattle, WA, May 2015, pp.
  4807--4812.

\bibitem{InfoTheoryBook}
T.~Cover and J.~Thomas, \emph{{Elements of Information Theory}}.\hskip 1em plus
  0.5em minus 0.4em\relax John Wiley \& Sons, 2012.

\bibitem{Ross2008}
\BIBentryALTinterwordspacing
S.~Ross, J.~Pineau, S.~Paquet, and B.~Chaib-draa, ``{Online planning algorithms
  for POMDPs},'' \emph{Journal of Artificial Intelligence Research}, vol.~32,
  no.~1, pp. 663--704, 2008. [Online]. Available:
  \url{http://www.aaai.org/Papers/JAIR/Vol32/JAIR-3217.pdf}
\BIBentrySTDinterwordspacing

\bibitem{Araya2010}
M.~Araya-L{\'o}pez, O.~Buffet, V.~Thomas, and F.~Charpillet, ``{A POMDP
  Extension with Belief-dependent Rewards},'' in \emph{Advances in Neural
  Information Processing Systems 23}, J.~Lafferty, C.~Williams,
  J.~Shawe-Taylor, R.~Zemel, and A.~Culotta, Eds., Vancouver, Canada, Dec.
  2010, pp. 64--72.

\bibitem{Silver2010}
D.~Silver and J.~Veness, ``{Monte-Carlo Planning in Large POMDPs},'' in
  \emph{Advances in Neural Information Processing Systems 23}, J.~Lafferty,
  C.~Williams, J.~Shawe-Taylor, R.~Zemel, and A.~Culotta, Eds., Vancouver,
  Canada, Dec. 2010, pp. 2164--2172.

\bibitem{Somani2013}
\BIBentryALTinterwordspacing
A.~Somani, N.~Ye, D.~Hsu, and W.~S. Lee, ``{DESPOT: Online POMDP Planning with
  Regularization},'' in \emph{Advances in Neural Information Processing Systems
  26}, C.~J.~C. Burges, L.~Bottou, M.~Welling, Z.~Ghahramani, and K.~Q.
  Weinberger, Eds.\hskip 1em plus 0.5em minus 0.4em\relax Curran Associates,
  Inc., 2013, pp. 1772--1780. [Online]. Available:
  \url{http://papers.nips.cc/paper/5189-despot-online-pomdp-planning-with-regularization.pdf}
\BIBentrySTDinterwordspacing

\bibitem{Lovejoy1987}
W.~S. Lovejoy, ``{Some Monotonicity Results for Partially Observed Markov
  Decision Processes},'' pp. 736--743, 1987.

\bibitem{Rieder1991}
\BIBentryALTinterwordspacing
U.~Rieder, ``{Structural results for partially observed control models},''
  \emph{ZOR Zeitschrift f\"{u}r Operations Research Methods and Models of
  Operations Research}, vol.~35, no.~6, pp. 473--490, Nov. 1991. [Online].
  Available: \url{http://link.springer.com/10.1007/BF01415990}
\BIBentrySTDinterwordspacing

\bibitem{Krishnamurthy2015}
\BIBentryALTinterwordspacing
V.~Krishnamurthy and U.~Pareek, ``{Myopic Bounds for Optimal Policy of POMDPs:
  An Extension of Lovejoy’s Structural Results},'' \emph{Operations
  Research}, vol.~63, no.~2, pp. 428--434, Apr. 2015. [Online]. Available:
  \url{http://pubsonline.informs.org/doi/10.1287/opre.2014.1332}
\BIBentrySTDinterwordspacing

\bibitem{Krishnamurthy2009}
V.~Krishnamurthy and B.~Wahlberg, ``{Partially Observed Markov Decision Process
  Multiarmed Bandits--Structural Results},'' \emph{Mathematics of Operations
  Research}, vol.~34, no.~2, pp. 287--302, 2009.

\bibitem{Krishnamurthy2007}
\BIBentryALTinterwordspacing
V.~Krishnamurthy and D.~V. Djonin, ``{Structured threshold policies for dynamic
  sensor scheduling - A partially observed Markov decision process approach},''
  \emph{IEEE Transactions on Signal Processing}, vol.~55, no.~10, pp.
  4938--4957, Oct. 2007. [Online]. Available:
  \url{http://ieeexplore.ieee.org/lpdocs/epic03/wrapper.htm?arnumber=4305446}
\BIBentrySTDinterwordspacing

\bibitem{Naghshvar2013}
\BIBentryALTinterwordspacing
M.~Naghshvar and T.~Javidi, ``{Active sequential hypothesis testing},''
  \emph{The Annals of Statistics}, vol.~41, no.~6, pp. 2703--2738, Dec. 2013.
  [Online]. Available: \url{http://projecteuclid.org/euclid.aos/1387313387}
\BIBentrySTDinterwordspacing

\bibitem{Naghshvar2010}
\BIBentryALTinterwordspacing
------, ``{Active M-ary sequential hypothesis testing},'' in \emph{2010 IEEE
  International Symposium on Information Theory}, vol.~1, no.~1.\hskip 1em plus
  0.5em minus 0.4em\relax IEEE, Jun. 2010, pp. 1623--1627. [Online]. Available:
  \url{http://ieeexplore.ieee.org/lpdocs/epic03/wrapper.htm?arnumber=5513381}
\BIBentrySTDinterwordspacing

\bibitem{DeGroot2004}
M.~H. DeGroot, \emph{Optimal Statistical Decisions}.\hskip 1em plus 0.5em minus
  0.4em\relax Hoboken, New Jersey: John Wiley \& Sons, Inc., 2004, {Wiley
  Classics Library edition}.

\bibitem{Thrun2006}
S.~Thrun, W.~Burgard, and D.~Fox, \emph{{Probabilistic Robotics}}.\hskip 1em
  plus 0.5em minus 0.4em\relax Cambrdige, MA: The MIT Press, 2006.

\bibitem{Bertsekas1996}
D.~Bertsekas and S.~Shreve, \emph{Stochastic Optimal Control: The Discrete-Time
  Case}.\hskip 1em plus 0.5em minus 0.4em\relax Athena Scientific, 1996.

\bibitem{Shaked2007}
\BIBentryALTinterwordspacing
M.~Shaked and J.~G. Shanthikumar, \emph{{Stochastic Orders}}.\hskip 1em plus
  0.5em minus 0.4em\relax New York, NY: Springer New York, 2007. [Online].
  Available: \url{http://link.springer.com/10.1007/978-0-387-34675-5}
\BIBentrySTDinterwordspacing

\bibitem{Fallat2011}
S.~M. Fallat and C.~R. Johnson, \emph{Totally Nonnegative Matrices}.\hskip 1em
  plus 0.5em minus 0.4em\relax Princeton, NJ: Princeton University Press, 2011.

\bibitem{Karlin1968}
S.~Karlin, \emph{{Total Positivity}}.\hskip 1em plus 0.5em minus 0.4em\relax
  Stanford University Press, 1968.

\bibitem{copositive_check}
W.~Kaplan, ``\href{http://dx.doi.org/10.1016/S0024-3795(00)00138-5}{A Test for
  Copositive Matrices},'' \emph{Linear Algebra and its Applications}, vol. 313,
  no. 1-3, pp. 203--206, 2000.

\end{thebibliography}

\end{document}